\documentclass[10pt]{article}
\usepackage{amssymb,amsmath,amsthm,multirow}
\usepackage{color}
\usepackage[all]{xy}
\newtheorem{definition}{Definition}
\newtheorem{example}{Example}
\newtheorem{remark}{Remark}
\newtheorem{theorem}{Theorem}
\newtheorem{proposition}{Proposition}
\newtheorem{lemma}{Lemma}

\newcommand{\ii}{\mathbf{i}}
\newcommand{\XX}{\mathbf{X}}

\newcommand{\NN}{\mathbb{N}}
\newcommand{\ZZ}{\mathbb{Z}}
\usepackage{authblk}

\title{Codes over Affine Algebras with a Finite Commutative Chain coefficient Ring}
\author[1]{E. Mart\'{\i}nez-Moro\footnote{Partially funded by MINECO MTM2015-65764-C3-1-P.
}}
\author[1]{A. Pi\~nera-Nicol\'as\footnote{Partially supported by MINECO-13-MTM2013-45588-C3-1-P, and Principado de Asturias Grant GRUPIN14-142.}}
\author[2]{I.F. R\'ua$^\dag$}

\affil[1]{Institute of Mathematics (IMUVa),
Universidad de Valladolid\\
\tt{edgar.martinez@uva.es, alejandro.pinera@uva.es} }

\affil[2]{Departamento de Matem\'aticas, Universidad de Oviedo\\
\tt{rua@uniovi.es}}
\begin{document}

\maketitle

\begin{abstract} 

We consider codes defined over an affine algebra $\mathcal A=R[X_1,\dots,X_r]/\left\langle t_1(X_1),\dots,t_r(X_r)\right\rangle$, where $t_i(X_i)$ is a monic univariate polynomial over a finite commutative chain ring $R$. Namely, we study the $\mathcal A-$submodules of $\mathcal A^l$ ($l\in \mathbb{N}$). These codes generalize both the codes over finite quotients of polynomial rings and the multivariable codes over finite chain rings. {Some codes over Frobenius local rings that are not chain rings are also of this type}. A canonical generator matrix for these codes is introduced with the help of the Canonical Generating System. Duality of the codes is also considered.



\end{abstract}

\section{Introduction}

{Quasi-cyclic codes over a finite {commutative} chain ring $R$ can be represented as $(R[x]/\langle x^n-1\rangle)$-submodules of $(R[x]/\langle x^n-1\rangle)^l$,  generalizing the well known construction for finite fields in, for example, \cite{Cao2013}. For finite  {commutative} chain ring, the one-generator codes have been extensively studied (see, for example, the classical paper \cite{Norton-Salagean} and the references included in \cite{Gao2016}), whereas for finite fields the general situation was studied in \cite{Lally2001} and recently generalized in \cite{berger} to codes over finite quotients of polynomial rings, i.e., to $\mathbb F[x]/\langle f(x)\rangle$-submodules of $(\mathbb F[x]/\langle f(x)\rangle)^l$ where $l\in \mathbb{N}$ and $f(x)$ is a monic polynomial. Furthermore, Jitman and Ling studied quasi-abelian codes over finite fields using techniques based on the Discrete Fourier Transform. They also  gave a structural characterization, as well as an enumeration, of one-generator quasi-abelian codes in \cite{Jitman2015}.} 

{In this paper we will consider codes defined over an affine algebra $\mathcal A=R[X_1,\dots,X_r]/\left\langle t_1(X_1),\dots,t_r(X_r)\right\rangle$, where each $t_i(X_i)$ is a monic univariate polynomial over a finite commutative chain ring $R$, i.e., $\mathcal A-$submodules of $\mathcal A^l$. Therefore, this class of codes includes the codes defined in \cite{berger}  and also, when $l=1$, multivariable codes over finite {commutative} chain rings \cite{nuestro1,nuestro2}. The proposed approach,  which uses the concept of Canonical Generating System introduced in \cite{CGS}, allows the study of quasi-cyclic codes over a finite {commutative} chain ring and their multivariable generalizations in a broader polynomial way, that is, beyond the one-generator case. Notice that for several generators a trace representation can also be derived from the ideas in \cite{ling}, see for example Section~4 in \cite{Gao2016}.} {The former approach provides a way for defining codes over some Frobenius local rings which are not chain rings. Some of them, as can be seen in Examples~\ref{ex:1}}, \ref{Ex:Cod_1} and \ref{Ex:Cod_2}, have not been previously explored in the literature.

{The outline of the paper is as follows. In Section~\ref{sec:int} we state the basic facts on finite {commutative} chain rings and some examples of known families of codes that are included in our definition.  The construction of a canonical generator matrix for our codes is provided in Section~\ref{sec:cangen}.} {In the final section duality of these codes is considered.}

\section{Basic definitions and examples}\label{sec:int}

An associative, commutative, unital, finite ring $R$ is called \emph{chain ring} if it has a unique maximal ideal $M$ and it is principal (i.e, generated by an element $a$). This condition is equivalent \cite[Proposition 2.1]{CN-FCR} to the fact that the set of ideals of $R$ is the chain (hence its name) $\left\langle 0\right\rangle =\left\langle a^t \right\rangle \subsetneq \left\langle a^{t-1} \right\rangle \subsetneq \dots \subsetneq \left\langle a^1 \right\rangle= M \subsetneq \left\langle a^0 \right\rangle=R
$, where $t$ is the nilpotency index of the generator $a$. 
The quotient ring $\overline R=R/M$ is a finite field $\mathbb F_q$ where $q= p^d$ is a prime number power. Examples of finite {commutative} chain rings include Galois rings $GR(p^n,d)$ of characteristic $p^n$ and $p^{nd}$ elements (here $a=p$, and $t=n$) and, in particular, finite fields ($\mathbb{F}_q=GR(p,d)$) \cite{McD,bini}. Any element $r\in R$, can be written as $r=a^ir'$, where $0\le i\le t$ and $0\not=\overline{r'}\in\mathbb{F}_q$. The exponent $i$ is unique, and it is called the \emph{norm} of $r$, written $\|r\|$.

\emph{Affine algebras} with a finite commutative chain coefficient ring are quotients
$\mathcal A=R[X_1,\dots,X_r]/\left\langle t_1(X_1),\dots,t_r(X_r)\right\rangle
$, where $t_j(X_j)\in R[X_j]$ are monic polynomials of degree $n_j$. (i.e., {$|\mathcal A|=|R|^n$, where $n:=\prod_{j=1}^rn_j$}). Let us denote every element $f$ in the affine algebra $\mathcal A$ as $f=\sum f_{\ii}\XX^{\ii}$, where $\ii=(i_1,\dots,i_r)\in \mathbb{N}_0^r$, $f_\ii\in R$ and $\XX^\ii=X_1^{i_1}\dots X_r^{i_r}$. If we require $0\le i_j<n_j$ for all $1\le j\le r$, then the expresion is unique. We implicitely require this condition through the paper and identify the element $f\in\mathcal A$ with the polynomial $\sum f_{\ii}\XX^{\ii}\in R[X_1,\dots,X_r]$.

\begin{definition}
Let $l$ be a natural number. A \emph{linear code of index $l$ over the affine algebra $\mathcal A$} (an \emph{$\mathcal{A}-$code}) is an $\mathcal A-$submodule $C$ of the direct product $E=\mathcal A^l$. 
\end{definition}

For all $1\le i\le l$, we shall denote by $\pi_i(f_1,\dots,f_l)=f_i$ the canonical projection on the $i-$th component $\pi_i:E\to\mathcal A$.

\begin{example}\label{ex:1} We start giving some examples of codes that can be seen, according to our construction, as codes over affine algebras:
\begin{enumerate}
\item Ideals of $\mathcal A$ (i.e., linear codes of index $l=1$) are the \emph{multivariable codes} introduced in \cite{nuestro1,nuestro2}. As particular cases we have the multivariable codes over a finite field $\mathbb{F}_q$ \cite{Poli}, and well-known families of codes over a finite chain ring alphabet, such as \emph{cyclic} {($r=1,t_1(X_1)=X_1^{n_1}-1$), \emph{negacyclic} ($r=1,t_1(X_1)=X_1^{n_1}+1$), \emph{constancyclic} ($r=1,t_1(X_1)=X_1^{n_1}+\lambda$), \emph{polycyclic} ($r=1$) and \emph{abelian codes} ($t_i(X_i)=X_i^{n_i}-1,\forall i=1,\dots,r$) \cite{CN-FCR,SergioSteve}. 
\item \emph{Quasi-cyclic} codes over a finite commutative chain ring are obtained when $r=1,t_1(X_1)=X_1^{n_1}-1$ and $l>1$. {See for example \cite{Cao2013, Gao2016} and the references therein.}}
\item {\emph{Quasi-abelian} codes over a finite field are obtained when $r>1$ and $t_i(X_i)=X_1^{n_i}-1$, $i=1,\ldots, r$ and $l>1$  \cite{Jitman2015}.} 
\item The \emph{codes over finite quotients of polynomial rings} studied in \cite{berger} are particular cases when the coefficient ring $R$ is the finite field $\mathbb{F}_q$ and the polynomial ring is univariate, i.e., $r=1$.
\item Codes over non-chain rings of order 16, $\mathbb F_2[u,v]/\langle  u^2,v^2\rangle$ and $\mathbb Z_4[x]/\langle x^2\rangle$ of \cite{Martinez-Moro2015} and codes over the family of rings $R_{k}= \mathbb F_2[u_1,\ldots,u_k]/\langle  u_1^2,\ldots,u_k^2\rangle$ studied in \cite{DOUGHERTY2011205}.
\end{enumerate}
\end{example}

\begin{remark}
Notice that some kind of codes over Frobenius local rings can be described using codes over affine algebras. In particular, in some of them the alphabet is not a chain ring and the usual matricial approach over those rings can not be used. This is the case of the \emph{linear codes over the ring $\mathbb{Z}_4[x]/\langle x^2+2x\rangle$} studied by Mart\'\i nez-Moro, Szabo and Yildiz in \cite{MoroSzaboYildiz}, which can be obtained when $R=\mathbb{Z}_4,r=1,t_1(X_1)=X_1^2+2X_1$ and $l>1$. Some of these codes give Gray images which are extremal Type \emph{II} $\mathbb{Z}_4-$codes, as it can be seen in Example \ref{Ex:3.9_MMSY}.
\end{remark}

\begin{remark}
Observe that any $\mathcal A-$code $C$ is also an $R-$module. However, $C$ is non necessarily free neither as $\mathcal A-$module nor as $R-$module. This makes a difference with codes over the finite quotients of (univariate) polynomial rings (over finite fields),  which are always $\mathbb{F}_q-$vector spaces \cite[page 166]{berger}. 
\end{remark}

\begin{definition}\label{def:Phi}
Let $\preceq$ be an admissible monomial order on $\mathbb{N}_0^r$, i.e., a
monomial order such that $\mathbf{0}\preceq
\mathbf{i}$ and 
$\mathbf{i}+\mathbf{k}\preceq
\mathbf{j}+\mathbf{k},$ for all
$\mathbf{i},\mathbf{j},\mathbf{k}\in\NN_0^r$ with $\mathbf{i}\preceq \mathbf{j}$ \cite{CGS}. For any $f\in \mathcal{A}$ we can associate an element $\phi(f)\in R^{n}$ by listing its coefficients in the order given by $\preceq$. The map $\phi:\mathcal{A}\to R^n$ can be extended to $\Phi:E\to R^{nl}$ coordinatewise, and the \emph{$R-$image} of an $\mathcal{A}-$code $C$ is defined as the $R-$linear code $\Phi(C)\subseteq R^{nl}$.
\end{definition}

From now on, we shall assume that {an} admissible monomial order $\preceq$ is fixed (e.g., the lexicographic, graded lex or graded reverse lex orders \cite{cox}). If $w$ is a weight function in the ring $R$ (for instance, the Hamming weight when $R=\mathbb{F}_q$ or the Lee weight if $R=\ZZ/4\ZZ$), then it can be used to induce the weight $w^{nl}:E\to \mathbb{R}$, given by $w^{nl}(f_1,\dots,f_l)=\sum_{j=1}^{nl}w(\Phi(f_1,\dots,f_l)_j)$ \cite[Section 3]{Nechaev99}. 

It is natural to expect that the codes under study are generally not generated by a single codeword (observe that repeated-root multivariable codes are particular examples \cite{nuestro2,nuestro6}). So, following \cite{berger}, generators of an $\mathcal{A}-$code $C$ will be given in terms of a generator matrix.
 
\begin{definition}
Let $C$ be an $\mathcal{A}-$code, we shall say that $G\in\mathcal{A}^{k\times l}$ ($k\in\mathbb{N}$) is a \emph{generator matrix} of $C$ if the rows of $G$ generate $C$ as an $\mathcal{A}-$submodule of $E$, i.e., $C=\left\langle g^{(1)},\dots,g^{(k)}\right\rangle_{\mathcal{A}}=\{\sum_{i=1}^kf_ig^{(i)}\ |\ f_i\in\mathcal{A}\}$ (with $fg=(fg_1,\dots,fg_l)\in E$). A matrix $G\in \mathcal{A}^{K\times l}$ ($K\in \mathbb{N}$) will be called \emph{generator matrix over $R$} if the rows of $G$ generate $C$ as an $R-$submodule of $E$, i.e., $C=\left\langle g^{(1)},\dots,g^{(K)}\right\rangle_{R}=\{\sum_{i=1}^Kr_ig^{(i)}\ |\ r_i\in R\}$ (with $rg=(rg_1,\dots,rg_l)\in E$).
\end{definition}

It is clear that from every generator matrix $G$ of $C$, a generator matrix over $R$ can be obtained by a substitution of every row $g^{(i)}$ of $G$ by the set of rows $\{X_1^{m_1}\dots X_r^{m_r}g^{(i)}\ |\ 0\le m_i<n_i\}$.

\section{A canonical generator matrix for $\mathcal{A}-$codes}\label{sec:cangen}

In this section we show that any $\mathcal{A}-$code can be described by a generator matrix $G$ in a canonical form which generalizes that of \cite[Section 2]{berger} for codes over finite quotients of (univariate) polynomial rings (over finite fields), and that of \cite[Section 4]{nuestro2} for repeated-root multivariable codes over finite commutative chain rings. As in the latter, the \emph{Canonical Generating
Systems (CGS)} introduced in \cite{CGS} is the main tool in our study. We refer the reader to such a reference for details. 

Recall that in this paper we implicitely identify the elements in $\mathcal{A}$ with multivariate polynomials {$\sum f_{\ii}\XX^{\ii}$ with $0\le i_j<n_j$} for all $1\le j\le r$. If $f=\sum f_{\ii}\XX^{\ii}\in \mathcal{A}
$ and $\ii\in \mathbb{N}_0^r$, then $\hbox{Cf}(f,\XX^\ii)$ denotes the coefficient $f_\ii$. The set $\mathfrak{S}(f)=\{\ii\in\mathbb{N}_0^r\ |\ f_\ii\not=0\}$ is the \emph{support} of $f$, and the \emph{norm} of $f$ is $\|f\|=\max_{\ii}\{\|f_\ii\|\}$. The \emph{leading degree} of $f$, $\hbox{ldg}(f)$, is the maximal $\ii$ (w.r.t. $\preceq$) such that $\|f_\ii\|=\|f\|$. The \emph{leading term} (resp. \emph{coefficient, monomial}) is defined as $\hbox{Lt}(f)=\XX^{\hbox{ldg}(f)}$ (resp. $\hbox{Lc}(f)=\hbox{Cf}(f,\hbox{Lt}(f)),\hbox{Lm}(f)=\hbox{Lc}(f)\hbox{Lt}(f)$).

If $f,h,g\in \mathcal{A}$, we say that $f$ is \emph{$L-$reduced to $h$ mod $g$} if there exists $r\in R$ with the condition $0\not=r{\hbox{Lc}(g)}=\hbox{Cf}(f,\XX^\ii\hbox{Lt}(g))$, such that $h=f-r\XX^\ii g$. In general, given a subset $\chi\in \mathcal{A}$, we say that $f$ is \emph{$L-$reduced to $h$ mod $\chi$} if there exists a sequence of {$L-$reductions} of $h_i$ to $h_{i+1}$ mod $g_i\in \chi$ ($i=1,\dots,w-1$) with $h_0=f,h_w=h$.

The polynomial $f\in\mathcal{A}$ is called \emph{normal mod $\chi\subseteq \mathcal{A}$}, if it can not be $L-$reduced mod $\chi$. The set of all {polynomials} normal mod $\chi$ is denoted $N(\chi)$. The polynomial $h\in\mathcal{A}$ is called a \emph{normal form} of $f$ mod $\chi$, if $h\in N(\chi)$ and $f$ is $L-$reduced to $h$ mod $\chi$. The set of all normal forms of $f$ mod $\chi$ is denoted $NF(f,\chi)$.

A
finite subset $\mathcal F\subset \NN_0^r$ is called \emph{Ferrers
diagram} if, given $\textbf{s}\in\mathcal F$, we have
$\mathbf{u}\in\mathcal F$ for all {$\mathbf{u}\le 
\mathbf{s}$ (i.e., $u_j\le s_j$, for all $1\le j\le r$)}. If $\mathcal {KF}$ denotes the set of
minimal elements in the partially ordered set {$(\NN_0^r\setminus
\mathcal F,\le)$}, then an \emph{$\mathcal F$-monic} polynomial
is a polynomial of the form
$
g^{\mathbf{s}}= \XX^\mathbf{s}-\sum_{\mathbf{u\in
\mathcal{F}}}g_{\mathbf{u}}\XX^{\mathbf{u}},
$
where
$\mathbf{s}\in \mathcal{KF}$. A set of $\mathcal F$-monic
polynomials of the form 
$
\chi=\{g^{\mathbf{s}}\ |\
\mathbf{s}\in \mathcal{KF}\}
$
 is called a
\emph{characteristic set}. If, besides, {$|\chi|=|\overline{\chi}|$ and} the set $\overline \chi$ is a reduced
Gr\"obner basis of the ideal $\left\langle \overline
\chi\right\rangle\subseteq \mathbb F_q[X_1,\dots,X_r]$, then $\chi$ is called a
\emph{Krull system}.

As an application of
\cite[Theorem 4.3]{CGS} to affine algebras, we have the following result, which generalizes \cite[Theorem 1]{berger} and also \cite[Theorem 2]{nuestro2}.

\begin{theorem}\label{matriz}
Let $C$ be a nonzero $\mathcal{A}-$code, and let
$\preceq$ be a fixed monomial order {on} $\NN_0^r$. Then, there exist:
\begin{enumerate}
    \item A natural number $1\le k\le l$, and two sequences of $k$ natural numbers ($1\le i_1<\dots<i_k\le l$ and $0\le x_1,\dots,x_k\le t-1$);
    \item $k$ strictly decreasing chain
of Ferrers diagrams $\mathcal F_0^j\supset
\mathcal F_1^j\supset \dots\supset \mathcal F_{x_j}^j$ and $k$ sequences of $x_j+1$ natural numbers 
$0=b_0^j<b_1^j<\dots<b_{x_j}^j<b_{x_j+1}^j=t$ (for all $1\le j\le k$);
\end{enumerate}
 such that for all $1\le j\le k$:
\begin{enumerate}
    \item {If $(f_1,\dots,f_l)\in C$, then for all $0\le m\le x_j$ such that
    $\mathfrak S(f_{i_j})\subseteq \mathcal F_m^j$, we have $\|f_{i_j}\|\ge b_{m+1}^j$;}
    \item There exists, for all $0\le m\le x_j$, a Krull system $\chi_m^j$
    of $\mathcal F_m^j$-monic polynomials such that $a^{b_m^j}\mathcal
    F_m\subseteq \pi_{i_j}(C_{i_j})$, where $C_{i_j}$ is the subcode of $C$ whose codewords have the first $i_j$ components equal to zero;
    \item The set $\chi^j=\chi_0^j\cup
a^{b_1^j}\chi_1^j\cup\dots\cup a^{b_{x_j}^j}\chi_{x_j}^j$ is a \emph{CGS} of $\pi_{i_j}(C_{i_j})$, i.e.,
\begin{enumerate}
\item The projection $\pi_{i_j}(C_{i_j})$ is the $\mathcal{A}-$submodule (i.e., the ideal of $\mathcal{A}$) generated by $\chi^j$;
\item {$f_{i_j}\in \pi_{i_j}(C_{i_j})$ if and only if the recurring sequence given by $h_0=f,h_{m+1}=NF(h_m,\chi^j_m)$ satisfies $\|h_m\|\ge b_m^j$, for all $0\le m\le x_j$, and $h_{m+1}=0$;}
    \item If we define $|\mathcal F_{-1}^j|=n(=\prod_{j=1}^r n_j)$, then
    $$|\pi_{i_j}(C_{i_j})|=q^{\sum_{m=0}^{x_j}(t-b_m^j)(|\mathcal F_{m-1}^j|-|\mathcal F_{m}^j|)}$$
    \item The code $C$ has $q^{\sum_{j=1}^k\sum_{m=0}^{x_j}(t-b_m^j)(|\mathcal F_{m-1}^j|-|\mathcal F_{m}^j|)}$ codewords.
    \end{enumerate}
    \item For all $0\le m\le x_j,g\in \chi^j_m$ and $z\in\{i_j,\dots,l\}$, there exist elements $h_z^{j,m,g}\in\mathcal{A}$ such that
    for each $s\in\{i_{j+1},\dots, i_k\}$ the polynomial $h_s^{j,m,g}$ is a normal form mod $\chi^s$ and the matrix with rows
    $$\overrightarrow{g_h}=\left(0\ ,\ {\dots}\ ,\ 0\ ,\ \stackrel{{(i_j}}{a^{b^j_m}g}\ ,\ h_{i_j+1}^{j,m,g}\ ,\ \dots\ ,\ h_{l}^{j,m,g}\right)$$
    is a generator matrix of $C$.
\end{enumerate}
\end{theorem}

Before the proof, let us introduce an example to illustrate the theorem.

\begin{example}
Consider the ${\mathbb{Z}_8}[X_1,X_2]/\left<X_1^4,X_2^4\right>-$code $C$ of index $l=3$ generated by the following codewords:
$$\left\{(X_1^2,X_1^2,4X_2),(X_2^2,X_2^2,0),(X_1X_2,X_1X_2,4X_1),(2X_2,2X_2,0),(4X_1,4X_1,0)\right\}.$$
Let us obtain its canonical generator matrix w.r.t. the lexicographic monomial order. We begin with the computation of the CGS of the (ideal in ${\mathbb{Z}_8}[X_1,X_2]$ corresponding to the) projection $\pi_1(C)$, i.e., $\left<X_1^4,X_2^4,X_1^2,X_2^2,X_1X_2,2X_2,4X_1\right>$. It is $\chi^1=\chi^1_0\cup 2\chi^1_1\cup 4\chi^1_2$, where $\chi^1_0=\{X_1^2,X_2^2,X_1X_2\}$, $\chi^1_1=\{X_1^2,X_2\}$, $\chi^1_2=\{X_1,X_2\}$. And so the corresponding Ferrers diagrams are $\mathcal F_2^1=\{(0,0)\}\subseteq \mathcal F_1^1=\mathcal F_2^1\cup \{(1,0)\}\subseteq \mathcal F_0^1=\mathcal F_1^1\cup \{(0,1)\}$. Therefore, $$|\pi_1(C)|=2^{(3-0)(16-3)+(3-1)(3-2)+(3-2)(2-1)}=2^{42}.$$
Next, as $C_2=C_3$ (any codeword with a zero as first component {begins} with $(0,0)$) we proceed with the computation of the 
CGS of the (ideal in ${\mathbb{Z}_8}[X_1,X_2]$ corresponding to the) projection $\pi_3(C_3)$, i.e., $\left<X_1^4,X_2^4, {4(X_1^2+X_2^2),4X_1X_2,4X_2^3}\right>$. It is $\chi^2=\chi^2_0\cup 4\chi^2_1$, where $\chi^2_0=\{X_1^4,X_2^4\}$, $\chi^2_1=\{X_1^2+X_2^2,X_1X_2,X_2^3\}$. And so the corresponding Ferrers diagrams are $\mathcal F_2^1=\{(0,0),(1,0),(0,1),(0,2)\}\subseteq \mathcal F_1^0={\{0,1,2,3\}\times \{0,1,2,3\}}$. Therefore, $$|\pi_2(C)|=2^{(3-0)(16-16)+(3-2)(16-4)}=2^{12}.$$
Hence, the code has $2^{54}$ codewords and its canonical generator matrix is:
$$\left(\begin{array}{cc|c}X_1^2&X_1^2&4X_2\\X_2^2&X_2^2&0\\X_1X_2&X_1X_2&4X_1\\2X_1^2&2X_1^2&0\\2X_2&2X_2&0\\4X_1&4X_1&0\\4X_2&4X_2&0\\\hline0&0&4(X_1^2+X_2^2)\\0&0&4X_1X_2\\0&0&4X_2^3\end{array}\right).$$
Observe that the polynomials $X_1^4,X_2^4$ in the CGS do not correspond to any row in the matrix, as they are reduced to zero mod the defining ideal of $\mathcal A$. Notice also that the elements in the upper right corner are normal forms mod $\chi^2$.
\end{example}

\begin{proof}
For all $i=1,\dots,l+1$, let $C_i=\cap_{j=1}^{i-1}\pi_j^{-1}(0)$ be the subset of $C$ consisting of those codewords of the form $(0,\stackrel{(i-1}{\dots},0,g_{i},\dots,g_l)$. It is straightforward to check that $C=C_1\supseteq C_2\supseteq \dots \supseteq C_{l+1}=\{0\}$ is a decreasing chain of $\mathcal{A}-$(sub)codes of $C$. Let us refine this chain by preserving only those nonzero $C_i$ such that $C_i\not =C_{i+1}$, i.e., there exists a natural number $k\in\{1,\dots,l\}$ and a sequence of $k$ natural numbers $1\le i_1<\dots<i_k\le l$ such that 
$C=C_{i_1}\supsetneq C_{i_2}\supsetneq \dots \supsetneq C_{i_k}\not=\{0\}$ with the conditions $C_s=C_{i_{j+1}}$ for all $i_j<s\le i_{j+1}$ and for all $j\in\{1,\dots,k-1\}$ (i.e,. any codeword with its $s-1$ first components equal to zero, must have the following $i_{j+1}-s$ components equal to zero, too), $C_s=0$, for all $s>i_k$, and $C_s=C$, for all $s\le i_1$.

For each $1\le j\le k$, consider the projection $\pi_{i_j}(C_{i_j})\subseteq \mathcal{A}$. It is clear that it is an ideal of $\mathcal{A}$, because $C$ is an $\mathcal{A}-$code. Therefore, we can lift it to a unique ideal $I_{i_j}\lhd R[X_1,\dots,X_r]$ containing $\left\langle t_1(X_1),\dots, t_r(X_r)\right\rangle$ and such that $I_{i_j}/\left\langle t_1(X_1),\dots, t_r(X_r)\right\rangle\cong \pi_{i_j}(C)$. Therefore, we can apply \cite[Theorem 4.3]{CGS} to find the $k+1$ sequences of natural numbers $0\le x_1,\dots,x_k\le t-1$ and $0=b_0^j<b_1^j<\dots<b_{x_j}^j<b_{x_j+1}^j=t$ and the $k$ strictly decreasing chain
of Ferrers diagrams $\mathcal F_0^j\supset
\mathcal F_1^j\supset \dots\supset \mathcal F_{x_j}^j$  (where $1\le j\le k$) satisfying items \#1,\#2 and \#3 of the theorem.

Now, for all $1\le j\le k\ ,\  0\le m\le x_j$ and $g\in \chi^j_m$, since $a^{b^j_m}g\in a^{b_{m}^j}\chi_{m}^j\subseteq \chi^j\subseteq \pi_{i_j}(C_{i_j})$, there exist elements $w_z^{j,m,g}\in\mathcal{A}$ (for all $z\in\{i_j+1,\dots,l\}$) such that
$$\overrightarrow{g_w}=\left(0\ ,\ {\dots}\ ,\ 0\ ,\ \stackrel{{(i_j}}{a^{b^j_m}g}\ ,\ w_{i_j+1}^{j,m,g}\ ,\ \dots\ ,\ w_{l}^{j,m,g}\right)\in C.$$
We claim that the matrix $W$ consisting on all the elements of this form is a generator matrix of $C$. Namely, if $c$ is a nonzero codeword in $C$, then there exists $j\in\{1,\dots,k\}$ such that $c\in C_{i_j}\setminus C_{i_{j+1}}$ (where $C_{i_{k+1}}:=\{0\}$), i.e., the codeword has the form $(0,\stackrel{(i_j-1}{\dots},0,c_{i_j},\dots,c_l)$. Because $c_{i_j}\in \pi_{i_j}(C_{i_j})$ we can apply the $L-$reduction of $c_{i_j}$ mod $\chi^j$ to the element $c$ (simply substitute the elements $a^{b^j_m}g$ in $\chi^j$ by the corresponding row $\overrightarrow{g_w}$). The remaining word is in $C$, but since its $i_j-$th component is zero, it must be 
an element in $C_{s}$ where $s\in\{i_{j+1},\dots,i_{k+1}\}$. If $s=i_{k+1}$, then such a codeword is zero and we have shown that $c$ is generated by the rows of $W$. Otherwise, we can apply the same argument inductively to reduce the remaining codeword mod $\chi^{j+1},\dots,\hbox{ mod }\chi^k$, until the zero codeword is obtained.

Finally, the generator form stated in the theorem is obtained when we apply the $L-$reduction of $w_s^{j,m,g}$ mod $\chi^s$ to the corresponding rows of the matrix $W$ (where $s\in\{i_{j+1},\dots, i_k\},1\le j<k,0\le m\le x_j$ and $g\in \chi^j_m$).

\end{proof}

\begin{example}
{The (punctured) Generalized Kerdock Code \cite{KuzNec94} has a multivariable presentation over the Galois ring $R=GR(q^2,2^2)$ ($q=2^d$) of order $q^2$ and characteristic $2^2$, which can be seen as an $\mathcal{A}-$code in our setting (with $\mathcal{A}=R[X_1,\dots,X_{r}]/\left<X_1^{q^m-1}-1,X_2^2-1,\dots,X_{r}^2-1\right>,m\ge 3$ odd, $r=d+1$, and $l=1$) \cite{nuestro2}. We will show that the (extended) Generalized Kerdock Code admits also an $\mathcal{A}-$code presentation with $l>1$, and we will obtain its generator matrix.}
    
{Any element $b\in R$ can be decomposed
     as $b=\gamma_0(b)+2\gamma_1(b)$, where $\gamma_i(b)\in \Gamma(R)=\{b^q=b\ |\ b\in R\}$ (the \emph{Teichm\"uller Coordinate Set (TCS)} of $R$). The multiplicative group
$1+2R$ is a direct product
$\left\langle\eta_1\right\rangle\times\dots\times\left\langle\eta_d\right\rangle$
of $d$ subgroups of order $2$.
If $S=GR(q^{2m},2^2)$ is the Galois extension of odd degree $m$ over $R$, then the multiplicative group on nonzero elements in its TCS  $\Gamma(S)=\{b^{q^m}=b\ |\ b\in S\}$ is cyclic. Let
     $\theta\in \Gamma(S)$ be one of its generators (of order
     $\tau=q^m-1$). Consider the trace map $\emph{Tr}$ from $S$ onto $R$.}

{The (punctured) Generalized Kerdock Code is equivalent to the projection 
 of the $\mathcal
\mathcal{A}-code$ of index $l=1$ 
$$C=\left\{\left(\sum_{i_1=0}^{\tau-1}\sum_{i_2=0}^1\dots\sum_{i_r=0}^1\left((\emph{Tr}(\xi\theta^{i_1})+b)
\eta_1^{i_2}\dots\eta_l^{i_r}\right) X_1^{i_1}X_2^{i_2}\dots
X_r^{i_r}\right)\ |\ \xi\in S, b\in R\right\}$$
with the map 
$\gamma_1^{\tau q}=\stackrel{\tau q}{\overbrace{\gamma_1\times\dots\times
\gamma_1}}$,
i.e., $\gamma_1^{\tau q} (\Phi(C))\subseteq \Gamma(R)^{\tau q}$.
}

{For any codeword 
$$(f_{\xi,b}):=\left(\sum_{i_1=0}^{\tau-1}\sum_{i_2=0}^1\dots\sum_{i_r=0}^1\left((\emph{Tr}(\xi\theta^{i_1})+b)
\eta_1^{i_2}\dots\eta_d^{i_r}\right) X_1^{i_1}X_2^{i_2}\dots
X_r^{i_r}\right)$$ with $\xi\in S, b\in R$, 
observe that the parity-check sum of $$(\emph{Tr}(\xi)+b,\emph{Tr}(\xi\theta)+b,\dots, \emph{Tr}(\xi\theta^{\tau-1})+b,b)$$ is zero, and define
$$g_{\xi,b}:=\sum_{i_1=0}^{\tau-1}\sum_{i_2=0}^1\dots\sum_{i_r=0}^1\left(b
\eta_1^{i_2}\dots\eta_d^{i_r}\right) X_1^{i_1}X_2^{i_2}\dots
X_r^{i_r}.$$
Consider the
$\mathcal{A}-code$ $D=\left\{(f_{\xi,b},g_{\xi,b})\ |\ \xi\in S, b\in R\right\}$ of index $l=2$.
Its image $\gamma_1^{2\tau q}(\Phi(D))$, when punctured in the first $q^{m+1}$ components, is equivalent to the (extended) Generalized Kerdock Code. }

{Let us obtain the generator matrix of Theorem \ref{matriz} for $D$. From \cite{nuestro2} we know that the $\mathcal{A}-$code $C$ is 1-generated by
the codeword $$(H_0
P(X_1)Q(X_2,\dots,X_r))$$ where $H_0$ is the trailing coefficient of the unique irreducible divisor $H(X_1)\in R[X_1]$ of
$X_1^\tau-1$ such that $H(\theta)=0$,
$P(X_1)=\frac{X_1^{\tau}-1}{(X_1-1)H^*(X_1)}$ ($H^*(X_1)$ is
the reciprocal polynomial of $H(X_1)$) and
$$Q(X_2,\dots,X_r)=\sum_{i_2=0}^1\dots\sum_{i_r=0}^1(
\eta_1^{i_2}\dots\eta_d^{i_r}) X_2^{i_2}\dots
X_r^{i_r}.$$
Therefore, $D$ is also $1-$generated by construction, and its generator matrix is given by
$$\left(H_0
P(X_1)Q(X_2,\dots,X_r)\ ,\ \sum_{i_1=0}^{\tau-1}\frac{H_0}{H^*(1)}X_1^{i_1}Q(X_2,\dots,X_r)\right)$$
because the sum of {the} coefficients of the polynomial $H_0
P(X_1)$ is $\frac{-H_0}{H^*(1)}$.}

\end{example}

\section{Canonical generator matrix of the dual code}\label{sec:dual}

In this section we present some results about duality of codes. If $R$ is a finite commutative quasi-Frobenius ring,  duality for $R-$linear codes can be defined through the inner product of $R^{nl}$: $(r_1,\dots,r_{nl})\cdot (s_1,\dots,s_{nl})=\sum_{j=1}^{nl}r_js_j$. This notion can be translated into $\mathcal{A}-$codes using the map $\Phi$.

\begin{definition}
The {{$R$}}-\emph{dual code} of the $\mathcal{A}-$code $C$ is $$C^{\perp_R}=\{e\in E\ |\ \Phi(e)\cdot \Phi(c)=0\hbox{ , for all }c\in C\}.$$  
\end{definition}

In the special case of $\mathcal A=R[X_1,\dots,X_r]/\left\langle X_1^{n_1}-1,\dots,X_r^{n_r}-1\right\rangle$, it is customary to define an Hermitian inner product on $E$:  $\left\langle (f_1,\dots,f_l),(g_1,\dots,g_l)\right\rangle=\sum_{j=1}^l f_j\overline{g_j}$, where $\overline{\cdot}$ is a conjugation map on $\mathcal{A}$ defined by 

$$\begin{array}{ccc} \overline{\cdot}:\mathcal A &\longrightarrow& \mathcal A\\
\sum f_\ii \XX^\ii&\longrightarrow&\sum f_\ii \XX^{\mathbf{n}-\ii} \end{array},$$

\noindent with $\mathbf{n}=(n_1,\dots,n_r)$. The (usual) Euclidean inner product $\cdot$ on $R^{nl}$ and the Hermitian product $\langle\cdot,\cdot\rangle$ on $E$ can be related using the following result, whose proof can be found in Proposition 3.2 of \cite{ling}.

\begin{proposition}
Let $C$ be a code over $E=\mathcal A^l$. Let us denote by $\perp_E$ the dual taken with respect the Euclidean inner product $\cdot$ of $R^{nl}$ and $\perp_H$ the dual in $E$ taken with respect the Hermitian inner product $\langle\cdot,\cdot\rangle$. If $\Phi:E\rightarrow R^{nl}$ is the map of Definition \ref{def:Phi}, then $\Phi(C)^{\perp_E}=\Phi(C^{\perp_H}).$
\end{proposition}

However, in some cases, we have to consider $\mathcal A$-duality instead of $R$-duality. For example, this happens with codes over the Frobenius local ring $\mathbb Z_4[X]/\langle X^2\rangle$, since it is not possible to construct a duality preserving map from the ring $\mathcal A$ to $R^n$, as can be seen on \cite{MoroSzaboYildiz}. On the other hand, if there exists a map $\theta$ from the ring $\mathcal A$ to $R^n$ that preserves duality, then $\mathcal A$-duality can be analized as $R$-duality. This is the case of the following example which provides a code over an affine algebra whose Gray image has good properties. 

\begin{example}\label{Ex:3.9_MMSY}
Let $R=\mathbb{Z}_4$ and $t_1(X_1)=X_1^2+2X_1$ and let
$$A=\left[
\begin{array}{cccc}
2X_1&3+3X_1&3+3X_1&3+3X_1\\
3+3X_1&2X_1&3+3X_1&3+3X_1\\
3+3X_1&3+X_1&2X_1&1+3X_1\\
3+3X_1&1+3X_1&3+X_1&2X_1
\end{array}\right].$$

\noindent Then, $[I_4|A]$ is the generator matrix of a self-dual code $C$ over $R$ of length 8 such that its Gray image is a Type II extremal self-dual code over $\mathbb{Z}_4$ of length 16 (see  \cite{Martinez-Moro2015} and \cite{MoroSzaboYildiz} for details).  
\end{example}

If $\mathcal A$-duality is considered, the dual code is defined as $$C^{\perp_{\cal A}}=\left\{e\in E\ | e\cdot c=\sum_{i=1}^l e_ic_i=0 \hbox{ , for all }c\in C\right\}.$$
\noindent As the code $C$ can be seen as a subgroup of $(E,+)$, it can be proved that there exists a group isomorphism $C^{\perp_{\mathcal A}}\cong E/C$ and so, $|C^{\perp_{\mathcal A}}||C|=|E|$ (see \cite{Nechaev99}).

A matrix $H$ whose rows generate the dual code $C^{\perp_{\mathcal A}}$ as $\mathcal A$-module is known as \emph{parity-check matrix} of the code $C$. If such a matrix exists, then $C=\{c\in E\,|\, Hc^t=0  \}$, that is, $C=\ker H$. All linear codes over a quasi-Frobenius ring have a parity-check matrix $H$.

We will use the canonical generator systems (CGS) in order to find a parity-check matrix of a given code $C$. We will only consider codes over univariate polynomial rings $\mathcal A=R[X]/\langle t_1(X)\rangle$, with $t_1(X)$ monic and such that $\mathcal A$ is Frobenius. In the following, we will denote $\mathcal A$-duality as $\perp$.

Let $C$ be an $\mathcal A$-code of index $l$. As in Theorem \ref{matriz} let us consider the decreasing chain of $\mathcal{A}-$(sub)codes of $C$ given by $C=C_1\supseteq C_2\supseteq \dots \supseteq C_{l+1}=\{0\}$, with $C_i=\cap_{j=1}^{i-1}\pi_j^{-1}(0)$ for $i=1,\dots,l+1$. Let us denote by $C'$ the punctured code of $C_2$ on the first position. 

\begin{lemma}\label{Lemma:concat}
Let $c'$ be an element of $C'^{\perp}$. Then, there exists an element $c\in\mathcal A$ such that $(c,c')$ is an element of $C^{\perp}$.
\end{lemma}

\begin{proof}
Let us suppose that there exists at least one element of $C$ with $c_1\neq 0$. Otherwise, any element $c'\in C'^{\perp}$ can be extended to an element $(c_1,c')\in C^{\perp}$ by concatenation of any $c_1\in\mathcal A$, and the result follows trivially. 

Following Theorem \ref{matriz}, there exist $k$ Canonical Generating Systems (CGS) $\chi^{j}=\chi_0^j\cup a^{b_1^j}\chi_1^j\cup\ldots\cup a^{b_{x_j}^j}\chi_{x_j}^j$, with $0=b_0^j<b_1^j<\cdots<b_{x_j}^j<b_{x_{j+1}}^j=t$, such that, for $j=1,\ldots,k$, the set $\chi^j$ generates the projection $\pi_{i_j}(C_{i_j})$ as $\mathcal A$-submodule.

Let us define the map $\varphi:C^{\perp}\rightarrow C'^{\perp}$ such that for any $c=(c_1,c_2,\ldots,c_l)\in E$, $\varphi(c)=(c_2,\ldots,c_l)$. Since $c\in C^{\perp}$, clearly $c\cdot e=0$ for all $e\in C_2$, so $0=\sum_{i=2}^l c_ie_i$, i.e., $(c_2,\ldots,c_l)\in C'^{\perp}$ and $\varphi$ is well defined. Moreover, $\varphi$ is an $\mathcal A$-linear map and $\ker \varphi=\{(c_1,0,\ldots,0)\in E\,|\, c_1\cdot G_i=0 \mbox{ for all } G_i\in\chi^1 \}$. The ring $\mathcal A$ is Frobenius so, from Theorem \ref{matriz}, it follows that $|C'^{\perp}|=q^{(l-1)nt-\sum_{j=2}^k\sum_{m=0}^{x_k}(t-b_i^j)(|\mathcal F^j_{m-1}|-|\mathcal F^j_m|)}$, where $n$ is the degree of the polynomial $t_1(X_1)$ and $\mathcal F_m^j$ is the Ferrers diagram associated to $\chi^j_m$. On the other hand, also by Theorem \ref{matriz}, we have that $|\ker \varphi|= q^{nt-\sum_{m=0}^{x_1}(t-b_m^1)(|\mathcal F_{m-1}^1|-|\mathcal F_{m}^1|)}=|C^{\perp}|/|C'^{\perp}|$. Then, $\varphi$ is a surjective map and the result follows.

\end{proof}

\begin{theorem}
Let $C$ be an $\mathcal A$-code of index $l$ and let $H'$ be a generator matrix of the code $C'^{\perp}$. Then, there exist polinomials $G'_0,\ldots,G'_s$ and $h_{s+1},\ldots,h_k$ in $\mathcal A$ such that the matrix
$$\left[\begin{array}{c|ccc} G'_0&0&\cdots&0\\ \vdots&\vdots&\ddots&\vdots\\ G'_s&0&\ldots&0\\\hline h_{s+1}&&&\\\vdots&&H'&\\ h_k&&& \end{array}\right]$$
\noindent is a parity-check matrix of the code $C$.
\end{theorem}

\begin{proof}
From Theorem \ref{matriz}, we know that there exists a Canonical Generator System $\chi^1$ which generates the projection $\pi_1(C_1)$ as $\mathcal A$-submodule. Since $\mathcal A$ is a univariate polynomial ring, then $\chi^1=\{G_0,G_1,\ldots,G_{x_1}\}$ where $G_i$ are polynomials such that $||G_i||= b_i^1$. Now, for each $1\leq i\leq x_1$, let us denote $A_i=\{g\in\mathcal A\,|\, g\cdot G_i=0\}$, the annihilator of $G_i$ in $\mathcal A$. The intersection $A=\cap_{i=1}^{x_1} A_i$ is an ideal of $\mathcal A$ so, by \cite[Theorem 4.3]{CGS}, there exists $\chi'=\{G'_0,G'_1,\ldots,G'_s\}$, a Canonical Generator System of the ideal $A$. Since each row $(h'_{i2},\ldots,h'_{il})$, $s+1\leq i\leq k$, of $H'$ is an element of $C'^{\perp}$, by Lemma  \ref{Lemma:concat} it is possible to find a polinomial $h_i\in\mathcal A$ such that $(h_i,h'_{i2},\ldots,h'_{il})\in C^{\perp}$.

It is easy to see that {any row} of the matrix $H$ {is} in $C^{\perp}$. First of all, notice that the polynomials $G'_i$, $1\leq i\leq s$, generate $A$, so the row $(G'_i,0,\ldots,0)$ is an element of $C^{\perp}$. The other rows of $H$ are in $C^{\perp}$ by construction. Conversely, let $c=(c_1,c_2,\ldots,c_l)\in C^{\perp}$, then $c'=(c_2,\ldots,c_l)\in C'^{\perp}$ will be generated by the rows $h'^{(i)}$ of $H'$. That is, there exist elements $\beta_{s+1},\ldots,\beta_k\in\mathcal A$ such that $c'=\sum_{i=s+1}^k\beta_i h'^{(i)}$. If we denote by $(G_j,G_{j2},\cdots,G_{jl})$ the $j$-th row of the generator matrix $G$ of $C$, it is clear that $h_iG_j+\sum_{w=2}^l h'_{iw}G_{jw}=0$ for each $1\leq j\leq x_1$. So, for each $j$
$$\sum_{i=s+1}^k\beta_ih_i G_j =\sum_{i=s+1}^k\beta_i\left(-\sum_{w=2}^lh'_{iw}G_{jw}\right)=-\sum_{w=2}^lc_wG_{jw}=c_1G_j.$$
\noindent Then, we have that $(c_1-\sum_{i=s+1}^k\beta_ih_i)G_j=0$ for $j=1,\ldots,x_1$. That is, $c_1-\sum_{i=s+1}^k\beta_ih_i\in A$ and so, it can be written as an $\mathcal A$-linear combination of polynomials $G'_0,\ldots,G'_s$. The proof is complete. 
\end{proof}

We finish this section with two examples of parity-check matrices of codes over the rings $\mathbb{Z}_4[X]/\langle X^2+2X\rangle$ and $\mathbb{Z}_4[X]/\langle X^2\rangle$, respectively. Both rings are local and Frobenius but none of them is a chain ring. These rings are studied in \cite{Martinez-Moro2015}.
\begin{example}\label{Ex:Cod_1}
Let us consider the code $C$ over $\mathbb{Z}_4[X]/\langle X^2+2X\rangle$ of index 2 whose generator matrix is given by 
$$\left[\begin{array}{cc}X&0\\2&X \\0&2\end{array}\right].$$

\noindent The parity-check matrix of $C$ is calculated as follows. First of all, notice that $H'=\left[2\right]$. On the other hand, it is easy to see that $A_1=\mbox{Ann}(\langle X\rangle)=\langle X+2\rangle$ and $A_2=\mbox{Ann}(\langle 2 \rangle)=\langle 2\rangle$, so $A=\langle 2X\rangle$. The polynomial $h_2$ should satisfy the  conditions
$$\begin{array}{lcc}
h_2X+0&=&0,\\
h_22+2X&=&0
\end{array}$$
\noindent which imply $h_2=X+2$. Then, the parity-check matrix of $C$ will be
$$\left[\begin{array}{cc} 2X&0\\X+2&2\end{array}\right]\backsim \left[\begin{array}{cc} X+2&2\end{array}\right]. $$
\noindent So the dual code $C^{\perp}$ is generated by the single word $(X+2,2)$.

\begin{example}\label{Ex:Cod_2}
Let $C$ be the code of index 2 over $\mathbb{Z}_4[X]/\langle X^2\rangle$ whose generator matrix is the same one that in Example \ref{Ex:Cod_1}. In order to find the parity-check matrix of $C$, we notice that $H'=[2]$ and that $A_1=\mbox{Ann}(\langle X\rangle)=\langle X\rangle$, $A_2=\mbox{Ann}(\langle 2\rangle)=\langle 2\rangle$. Therefore $A=\langle 2X\rangle$. Finally, the polynomial $h_2$ satisfies the conditions 
$$\begin{array}{lcc}
h_2X+0&=&0,\\
h_22+2X&=&0
\end{array}$$ 
\noindent which {yield} to $h_2=X$. So, the parity-check matrix of $C$ is
$$\left[\begin{array}{cc} 2X&0\\X&2\end{array}\right]\backsim\left[\begin{array}{cc} X&2\end{array}\right].$$
\noindent Notice that the dual code $C^{\perp}$ is also generated by a single word: $(X,2)$. As it is referred in \cite{Martinez-Moro2015} there is no duality preserving map from the ring $\mathcal A=\mathbb{Z}_4[X]/\langle X^2\rangle$ to $\mathbb{Z}_4^2$. {So, $\mathcal A$-duality cannot be analized as $\mathbb{Z}_4-$duality.}
\end{example}
\end{example}

%


\def\cprime{$'$}

\end{document}